\documentclass[pra,aps,twocolumn,superscriptaddress,a4paper]{revtex4-1}

\usepackage[dvips]{graphicx}
\usepackage{mathtools}
\usepackage[british]{babel}
\usepackage{amsmath,amssymb,amsthm,dsfont,hyperref,ctable,url,commath}
\usepackage{color}
\usepackage{amsmath}
\usepackage{bbold}
\usepackage{braket, mathtools}
\usepackage{enumerate}

\newcommand{\ketbra}[2]{\ket{#1}\!\!\bra{#2}}
\newcommand{\kl} [2] {\operatorname{D}\left(#1\Big|\Big|#2\right)}
\newcommand{\conv}{\operatorname{conv}}
\newcommand{\hg}{ {{}_2 \! \operatorname{F}_1}}
\newcommand{\Tr}{\operatorname{Tr}}

\newtheorem{thm}{Theorem}
\newtheorem{lmm}{Lemma}

\newtheorem{dfn}{Definition}

\DeclarePairedDelimiter{\ceil}{\lceil}{\rceil}

\begin{document}

\title{Communication capacity of mixed quantum $t$ designs}

\author{Sarah \surname{Brandsen}}

\email{sbrandse@caltech.edu}

\affiliation{Centre    for   Quantum    Technologies,   National
  University of Singapore, 3  Science Drive 2, 117543 Singapore,
  Republic  of Singapore}

\author{Michele \surname{Dall'Arno}}

\email{cqtmda@nus.edu.sg}

\affiliation{Centre    for   Quantum    Technologies,   National
  University of Singapore, 3  Science Drive 2, 117543 Singapore,
  Republic of Singapore}

\author{Anna \surname{Szymusiak}}

\email{anna.szymusiak@uj.edu.pl}

\affiliation{Institute of  Mathematics, Jagiellonian University,
  \L ojasiewicza 6, 30-348 Krak\'ow, Poland}

\date{\today}

\begin{abstract}
  We operationally  introduce mixed  quantum $t$ designs  as the
  most  general arbitrary-rank  extension of  projective quantum
  $t$  designs  which  preserves indistinguishability  from  the
  uniform distribution  for $t$ copies.  First,  we derive upper
  bounds on  the classical  communication capacity of  any mixed
  $t$  design  measurement, for  $t  \in  [1, 5]$.   Second,  we
  explicitly  compute the  classical  communication capacity  of
  several   mixed  $t$   design   measurements,  including   the
  depolarized  version  of:  any  qubit  and  qutrit  symmetric,
  informationally  complete   (SIC)  measurement   and  complete
  mutually   unbiased  bases   (MUB),   the  qubit   icosahedral
  measurement, the  Hoggar SIC measurement, any  anti-SIC (where
  each element is proportional to  the projector on the subspace
  orthogonal to  one of the  elements of the original  SIC), and
  the uniform distribution over pure effects.
\end{abstract}

\maketitle

\section{Introduction}

Arguably, the class of  quantum configurations with the broadest
application in any communication  protocol is that of projective
$t$  designs~\cite{Sco06,  AE07},   including  e.g.   symmetric,
informationally  complete (SIC)  measurements~\cite{Zau99, KK01,
  RBSC04,  App05,  SG10,  Chi15}  and  mutually  unbiased  bases
(MUB)~\cite{KR05,  Ben07, BH07,  BW08,  BWB10}.  Projective  $t$
designs can informally be defined  by the properties of i) being
indistinguishable  from  the  uniform  distribution  up  to  $t$
copies,   and  ii)   comprising  only   pure  elements.    While
indistinguishability from  the uniform  distribution is  an {\em
  operational} restriction of the  protocol, introduced e.g.  to
ensure that  no information leaks  to an adversary, purity  is a
purely {\em mathematical}  idealization bound to be  lost in the
presence of noise.

Previous  generalizations  of  projective  $t$  designs  to  the
arbitrary-rank  case were  either limited  to specific  subsets,
such as SICs  or MUBs~\cite{App07, KG13, KG14,  Ras14, CS14}, or
relaxed  the operational  property of  indistinguishability from
the uniform distribution given  $t$ copies~\cite{GA15}.  In this
work, we  introduce the class of  mixed $t$ designs as  the most
general extension  of projective $t$ designs  that preserves the
operational  property of  indistinguishability from  the uniform
distribution up  to $t$ copies, while  relaxing the mathematical
assumption of pure elements.

In   any   communication   scenario,   such   as   quantum   key
distribution~\cite{BB84}     or    locking~\cite{DHLST04}     of
information in quantum states, a relevant figure of merit is the
classical communication  capacity, namely the maximum  amount of
information  that can  reliably  be extracted  per  use, in  the
asymptotic  limit.   More  precisely, the  celebrated  Shannon's
coding   theorem~\cite{Sha48}    guarantees   that,    for   any
communication channel,  there exist  an encoding and  a decoding
such that the amount of information that can be transmitted with
null error  probability is  equal to  the channel  capacity, and
also that no better performance can be achieved.

For measurements,  the capacity was  proven to be  equivalent to
its   single-shot  analogy,   i.e.   the   informational  power,
introduced in Ref.~\cite{DDS11}. Therein, the simple case of SIC
measurement of a qubit  was solved. Subsequently, considerable
interest arose  for computing the capacity  of several symmetric
measurements~\cite{OCMB11,  Hol12,  Hol13, SS14,  DBO14,  Szy14,
  Dal14,  Dal15,  SS15}.   With very  limited  exceptions  (e.g.
mirror symmetric measurements, introduced in Ref.~\cite{DDS11}),
mixed $t$  designs introduced  here (with $t  \ge 2$,  where our
results  provide novel  insights) encompass  all the  classes of
non-trivial  measurements for  which the  classical capacity  is
known, and extend such classes by adding e.g.  the isotropically
noisy version of projective $t$ designs.

Here, we introduce  mixed $t$ designs as the  most general class
of quantum  configurations that  are indistinguishable  from the
uniform distribution given $t$ copies, and discuss some of their
basic  properties.   First,  we   derive  upper  bounds  on  the
classical  communication  capacity  of   any  mixed  $t$  design
measurement, for $t \in [1,  5]$.  Second, we explicitly compute
the classical communication capacity of several mixed $t$ design
measurements, including  the depolarized  version of:  any qubit
and  qutrit   SIC  measurement  and  complete   MUB,  the  qubit
icosahedral  measurement,   the  Hoggar  SIC   measurement,  any
anti-SIC (where each element is proportional to the projector on
the subspace orthogonal  to one of the elements  of the original
SIC), and the uniform distribution over pure effects.

This work is structured as follows.  In Sec.~\ref{sec:design} we
introduce  mixed $t$  designs  and derive  some  of their  basic
properties.    In   Sec.~\ref{sec:capacity}  we   discuss   some
fundamental facts  about the  communication capacity  of quantum
measurements.   In Sec.~\ref{sec:results}  we  present our  main
results about  the capacity  of mixed $t$  designs measurements.
Finally, we summarize  our results and discuss  some outlooks in
Sec.~\ref{sec:conclusion}.

\section{Classical capacity of mixed $t$ designs}

\subsection{Mixed $t$ designs}
\label{sec:design}

Let us  first recall  some standard  definitions and  facts from
quantum  theory~\cite{NC00}.  Let  $\mathcal{H}$  be a  (finite)
$d$-dimensional Hilbert space, $\mathcal{L}(\mathcal{H})$ be the
space  of  linear  operators   on  $\mathcal{H}$  and  $\mathcal
L_+(\mathcal  H)$ be  the set  of all  the positive-semidefinite
operators  on  $\mathcal{H}$.   The most  general  quantum  {\em
  state}  in $\mathcal{H}$  is  described by  a density  matrix,
namely  a  trace-one operator  in  $\mathcal{L}_+(\mathcal{H})$.
Adopting  the Dirac  notation,  any pure  state  $\phi$ is  also
denoted by $\ketbra{\phi}{\phi}$ (thus $\phi$ corresponds to the
projector on $\ket{\phi}$).  Let $(A,\Sigma)$ be a Borel locally
compact metric  space equipped  with a  Borel measure  $\mu$ (In
particular,  one can  think of  a  finite set  with the  uniform
measure or a  locally compact group with the  Haar measure, e.g.
the group $U(d)$ of unitaries in $\mathcal L(\mathcal H)$.)  The
most general  quantum preparation in $\mathcal{H}$  is described
by  a  quantum  {\em  ensemble}, namely  a  measurable  function
$\rho:A\ni       x\mapsto\rho_x:=p_x\hat\rho_x\in       \mathcal
L_+(\mathcal  H)$, where  $\hat\rho_x$ is  a quantum  state, and
$x\mapsto p_x$ is a probability density (with respect to $\mu$).
The most  general quantum {\em measurement}  in $\mathcal{H}$ is
described  by a  positive  operator-valued  measure (POVM).   In
particular,  a POVM  can  be defined  by  a measurable  function
$\pi:A\ni  y\mapsto\pi_y:=dp_y\hat\pi_y\in\mathcal  L_+(\mathcal
H)$, where $\hat\pi_y$  is a quantum state, $y\mapsto  p_y$ is a
probability density  (with respect to  $\mu$), and $d  \int \dif
\mu(y)  p_y  \hat\pi_y =  \openone$,  the  identity operator  in
$\mathcal{L}(\mathcal{H})$.

We introduce in the  following some quantities that characterize
ensembles and  POVMs.  For  any measurable function  $\chi$ from
$A$ to  positive semidefinite  operators $x\mapsto\chi_x  := \nu
p_x \hat\chi_x \in \mathcal{L}(\mathcal{H})$  such that $\nu$ is
a  normalization, $x\mapsto  p_x$ is  a probability  density and
$\hat\chi_x$ a quantum state, we call
\begin{align*}
  \mu_k(\chi) := \int \dif \mu(x) p_x \Tr[\hat\chi_x^k],
\end{align*}
the $k$-th {\em moment} of $\chi$.  For any pure state $\phi \in
\mathcal{L}(\mathcal{H})$, we call
\begin{align*}
  \gamma_k(\chi, \phi) := \int \dif \mu(x) p_x \left( \bra{\phi}
    \hat\chi_x \ket{\phi} \right)^k,
\end{align*}
the {\em index of coincidence}.

Operationally, a {\em projective quantum $t$ design ensemble} is
an ensemble  $x\mapsto\phi_x :=  p_x \hat\phi_x$ of  pure states
that  cannot be  discriminated  from  the uniformly  distributed
ensemble of  pure states  -- other than  by trivial  guessing --
when given  $t$ copies.  Analogously, a  {\em projective quantum
  $t$  design   POVM}  is   operationally  defined  as   a  POVM
$y\mapsto\pi_y := q_y \hat\pi_y$  such that the ensemble $\phi_y
:=  \Tr_2[\Phi^+ (\openone  \otimes \pi_y)]  = \frac1d  \pi_y^T$
steered  on   a  maximally  entangled  state   $\ket{\Phi^+}  :=
\frac1{\sqrt{d}} \sum_i  \ket{i,i}$ is a projective  quantum $t$
design ensemble.  More generally,  a {\em projective quantum $t$
  design} can  be defined as  a measurable function  $\phi$ from
$A$ to rank-one  positive semidefinite operators $x\mapsto\phi_x
:= \nu p_x \hat\phi_x$ such that
\begin{align*}
  \int  \dif \mu(x)  p_x \hat\phi_x^{\otimes  t} =  \int \dif  g
  \ketbra{\phi_g}{\phi_g}^{\otimes t},
\end{align*}
where $\dif  g$ denotes  the uniform (Haar)  probability measure
over      the     group      $U(d)$     of      unitaries     in
$\mathcal{L}(\mathcal{H})$, see~\cite{Sco06}.

We introduce mixed $t$ designs as a generalization of projective
$t$  designs  that  relax   the  constraint  of  being  rank-one
projectors  while   preserving  indistinguishability   from  the
uniform distribution given $t$ copies.  Operationally, we define
{\em  mixed   $t$  design  ensembles}  $x\mapsto\rho_x   :=  p_x
\hat\rho_x$ as ensembles of  (possibly) mixed states that cannot
be distinguished from the uniform distribution given $t$ copies.
Analogously,  we  operationally  define {\em  mixed  $t$  design
  POVMs}  as POVMs  $y  \mapsto \pi_y$  such  that the  ensemble
$\rho_y  := \Tr_2[\Phi^+  (\openone  \otimes  \pi_y)] =  \frac1d
\pi_y^T$ steered on a maximally entangled state $\ket{\Phi^+} :=
\frac1{\sqrt{d}}  \sum_i  \ket{i,i}$  is   a  mixed  $t$  design
ensemble.   More generally,  we  define {\em  mixed quantum  $t$
  designs} as follows.

\begin{dfn}[Mixed $t$ design]
  \label{def:design}
  We call  a mixed  quantum $t$  design any  measurable function
  $\chi$   from   $A$   to   positive   semidefinite   operators
  $x\mapsto\chi_x :=  \nu p_x \hat\chi_x$, where  $x\mapsto p_x$
  is a probability density (with respect to $\mu$) if
  \begin{align}
    \label{eq:design}
    \int \dif \mu(x)  p_x \hat\chi_x^{\otimes t} :=  \int \dif g
    (U_g \hat\chi U_g^\dagger)^{\otimes t},
  \end{align}
  where $\dif g$  is the Haar probability measure  on $U(d)$ and
  $\hat\chi$ is some unit-trace positive-semidefinite operator.
\end{dfn}

In  the following,  for brevity  we refer  to mixed  $t$ designs
simply as $t$ designs. Let us prove some important properties of
$t$ designs. First, notice that if  $\chi$ is a $t$ design, also
$\chi^T$   is   a  $t$   design,   and   thus  the   operational
interpretation of $t$ design  POVMs immediately follows.  Notice
also   that   by   partial  tracing   Eq.~\eqref{eq:design}   it
immediately follows that any $t$ design is also a $k$ design for
any $1 \le k \le t$, and therefore in particular
\begin{align*}
  \int \dif\mu(x) p_x \hat\chi_x = \frac{\openone}d,
\end{align*}
for any $t$ design $\chi$.

The moments up to $t$ of any $t$ design $\chi$ are given by
\begin{align}
  \label{eq:moment}
  \mu_k(\chi) = \Tr[\hat\chi^k].
\end{align}
for any $1 \le k \le  t$. This uniquely identifies $\hat\chi$ up
to unitaries  if $t  \ge d$.  Equation~\eqref{eq:moment}  can be
readily verified  by first multiplying  Eq.~\eqref{eq:design} by
the shift operator $S_t  := \sum_{i_1,\dots i_t} \bigotimes_{j =
  1}^t \ket{i_{j\oplus1}}\bra{i_j}$  that replaces the  state of
$\mathcal{H}_{j\oplus1}$  with  $\mathcal{H}_j$  (here  $\oplus$
denotes  sum modulus  $t$ and  $\{ \ket{i_j}  \}_{i=1}^d$ is  an
orthonormal basis  of space  $\mathcal{H}_j$), and  then tracing
using  the   property  of  the  shift   operator  that  $\Tr[S_t
V^{\otimes  t}]  =  \Tr[V^t]$  for any  $t$  and  any  Hermitian
operator $V$.

For quantum $t$  designs $\chi$, it follows  from the definition
that for any $1 \le k \le t$ the index of coincidence $\gamma_k$
is independent of $\phi$ and is given by
\begin{align}
  \label{eq:coincidence}
  \gamma_k(\chi)   &   =   {d+k-1  \choose   k}^{-1}   \Tr[P_k
  \hat\chi^{\otimes k}] \nonumber\\  & = \frac{(d-1)!}{(d+k-1)!}
  B_k (x_1, \dots x_k),
\end{align}
where $P_k$ denotes the projector over the symmetric subspace of
$\mathcal{H}^{\otimes k}$, $x_i := (i-1)!  \Tr[\hat\chi^i]$, and
$B_k(x_1,  \dots   x_k)$  is   the  complete   exponential  Bell
polynomial~\cite{Bel27} given by
\begin{align*}
  B_k  (x_1, \dots,  x_k  ) =  \sum_{j=1}^k  B_{k,j} (x_1,  x_2,
  \dots, x_{k-j+1}),
\end{align*}
where
\begin{align*}
  &  B_{k,   j}(x_1,  x_2,  \dots,   x_{k-j+1})  \\  =   &  \sum
  \frac{k!}{i_{1}!    \cdots   i_{k-j+1}!}    \left({x_1   \over
      1!}\right)^{i_1}     \cdots     \left({x_{k-j+1}     \over
      (k-j+1)!}\right)^{i_{k-j+1}},
\end{align*}
and the sum  is over all sequences $i_1,  i_2, \dots, i_{k-j+1}$
such that
\begin{align*}
  \begin{cases}
    i_1 + i_2 + \cdots  + i_{k-j+1}  = j, \\
    i_1 + 2 i_2 + 3 i_3 + \cdots + (k-j+1) i_{k-j+1} = k.
  \end{cases}
\end{align*}
The second equality in Eq.~\eqref{eq:coincidence} follows from a
lengthy  but   straightforward  counting   argument.   Therefore
explicitly for $k \in [1,5]$ one has
\begin{align*}
  \begin{cases}
    \gamma_1 & = \frac1d,\\
    \gamma_2 & = \frac{1+\mu_2}{d(d+1)},\\
    \gamma_3 & = \frac{1+3\mu_2+2\mu_3}{d(d+1)(d+2)},\\
    \gamma_4 & = \frac{1 + 6\mu_2 + 3\mu_2^2 + 8\mu_3 + 6\mu_4}
    {d(d+1)(d+2)(d+3)},\\
    \gamma_5 & = \frac  {1 + 10 \mu_2 + 15 \mu_2^2  + 20 \mu_3 +
      30    \mu_4   +    20    \mu_2   \mu_3    +   24    \mu_5}
    {d(d+1)(d+2)(d+3)(d+4)}.
  \end{cases}
\end{align*}

Relevant   classes    of   $t$   designs    include   symmetric,
informationally  complete (SIC)  POVMs~\cite{Zau99, RBSC04}  and
mutually  unbiased bases  (MUB)~\cite{KR05},  as  well as  their
mixed   generalizations  given   by  symmetric   informationally
complete  measurements  (SIM)~\cite{App07,  KG13}  and  mutually
unbiased   measurements    (MUM)~\cite{KG14}.    Despite   their
importance, the existence of SICs and  MUBs is still a matter of
conjecture~\cite{App05}.   For SICs,  existence has  been proved
analytically for  dimensions $2-15$ and $17$,  $19$, $24$, $35$,
and   $48$~\cite{Zau99,   RBSC04,   KK01,  SG10,   Chi15},   and
numerically for dimensions $2-67$~\cite{RBSC04, SG10}. For MUBs,
existence   is  questioned   for   a  dimension   as  small   as
$6$~\cite{Ben07,BH07,BW08}. On the other  hand, it is known that
MUBs exist  in infinitely many dimensions,  i.e.  all dimensions
being   prime   powers~\cite{BBRV02}.    Other   than   in   the
aforementioned  communications  protocols,  SICs and  MUBs  have
applications  in quantum  tomography~\cite{DDPPS02}, uncertainty
relations~\cite{BHOW13},        and       in        foundational
problems~\cite{FS03,FS09,FS11,AEF11,Fuc12}.

Any  function $\chi$  from $\{0,1,\ldots,  d^2-1\}$ to  positive
semidefinite   operators   $x\mapsto\chi_x:  =   \frac{\nu}{d^2}
\ketbra{\hat\chi_x}{\hat\chi_x}$ such that
\begin{align*}
  | \braket{\hat\chi_x|\hat\chi_{x'}} |^2 = \delta_{x, {x'}} +
  (1-\delta_{x,{x'}}) \frac{1}{d+1},
\end{align*}
for any $x$ and $x'$ defines  a set $\{\chi_x\}_x$ which we call
a symmetric informationally complete  (SIC) set.  Any measurable
function  $\chi$  from  $\{0,1,\ldots,d\}  \times  \{0,1,\ldots,
d-1\}$       to       positive      semidefinite       operators
$(x,y)\mapsto\chi_{x,y}:           =          \frac{\nu}{d(d+1)}
\ketbra{\hat\chi_{x,y}}{\hat\chi_{x,y}}$ such that
\begin{align*}
  |    \braket{\hat\chi_{x,     y}|\hat\chi_{x',y'}}    |^2    =
  \delta_{x,x'}\delta_{y,y'}, + (1 - \delta_{x,x'}) \frac1d,
\end{align*}
for   any   $x$,   $y$,   $x'$,   and   $y'$   defines   a   set
$\{\chi_{x,y}\}_{x,y}$ which  we call  a mutually  unbiased set.
For simplicity,  we use  the notation  $\chi_{x,y} =  \chi_{dx +
  y}$.

An explicit way to construct families of mixed $t$ designs is by
the affine  combination of any mixed  $t$ design $x\mapsto\chi_x
:= p_x \hat\chi_x$ with the maximally mixed operator as follows:
\begin{align*}
  \mathcal{D}_\lambda(\chi_x)  =  \lambda \chi_x  +  (1-\lambda)
  \Tr[\chi_x] \frac{\openone}d,
\end{align*}
for any  $\lambda$ such that $\mathcal{D}_\lambda(\chi_x)  \ge 0
\;      \forall     x$,      as     follows      by     applying
$\mathcal{D}_\lambda^{\otimes    t}$    to   both    sides    of
Eq.~\eqref{eq:design}    and     noticing    that     the    map
$\mathcal{D}_\lambda$  commutes with  any unitary  channel. More
precisely,    let    $\hat\chi_x    :=    \sum_k    \alpha_{x,k}
\ketbra{\phi_{x,k}}{\phi_{x,k}}$ be a  spectral decomposition of
$\hat\chi_x$  and   let  $\alpha_{\textrm{max}}   :=  \max_{x,k}
\alpha_{x,k}$    and   $\alpha_{\textrm{min}}    :=   \min_{x,k}
\alpha_{x,k}$.   Then  by  direct  inspection  it  follows  that
$\mathcal{D}_\lambda(\chi_x) \ge 0$ if and only if
\begin{align}
  \label{eq:affine}
  \lambda  \in \left[  \frac1{1 -  d \alpha_{\textrm{max}}}  \le
    \frac1{1-d}, \frac1{1 - d \alpha_{\min}} \ge 1 \right].
\end{align}
The   linear   map   $\mathcal{D}_\lambda$  corresponds   to   a
depolarizing channel if  and only if $\lambda \in  [0, 1]$. Then
it follows that  the depolarized version of any $t$  design is a
$t$ design.  Furthermore,  for any $t$ design  $\chi$, we define
the corresponding  {\em anti  $t$ design}  as $\mathcal{D}_{(1-d
  \alpha_{\textrm{max}})^{-1}}  (\chi)$.   Notice  that  due  to
Eq.~\eqref{eq:affine}  any  anti $t$  design  is  a $t$  design.
Finally  one  has  the  following simple  relation  between  the
moments of  any $t$ design  $\chi$ and those of  its depolarized
version $\mathcal{D}_\lambda(\chi)$:
\begin{align*}
  \mu_k(\mathcal{D}_\lambda(\chi))  = \sum_{n=0}^k  {k \choose
    n}       \lambda^n      \left(\frac{1-\lambda}d\right)^{k-n}
  \mu_n(\chi).
\end{align*}

\subsection{Communication capacity}
\label{sec:capacity}

Let us  first recall  some standard  definitions and  facts from
information   theory~\cite{CT06}.   Intuitively,   a  means   of
quantifying  the   distinctiveness  of  two   given  probability
densities $p$ and $q$ on $A$ (with respect to $\mu$) is given by
the  {\em  relative  entropy}  $\kl{p}{q}$, also  known  as  the
Kullback-Leibler divergence, defined as
\begin{align*}
  \kl{p}{q} := \int\dif \mu(x) p_x \ln\frac{p_x}{q_x} .
\end{align*}
A  measure  of the  correlation  between  any two  given  random
variables  $X$  and  $Y$  distributed  according  to  the  joint
probability  density $p_{(X,Y)}$,  is given  by the  {\em mutual
  information} $I(X;Y)$ defined as
\begin{align*}
  I(X;Y) := \kl{p_{(X,Y)}}{p_X p_Y},
\end{align*}
where   $p_X$   and   $p_Y$   are   the   marginal   probability
distributions.  For any ensemble $\rho$ and POVM $\pi$ we denote
with  $I(\rho, \pi)$  the  mutual  information $I(X;Y)$  between
random   variables  $X$   and  $Y$   distributed  according   to
$(x,y)\mapsto p_{x,y} := \Tr[\rho_x \pi_y]$.

Following Ref.~\cite{Hol12}, we define  the capacity of any POVM
as follows.

\begin{dfn}[Unassisted classical capacity of POVMs]
  \label{def:povm-capacity}
  The unassisted classical  capacity of any POVM  $\pi$ is given
  by
  \begin{align*}
    C(\pi)   :=    \lim_{t\to\infty}   \frac1t   \max_{\rho}
    I(\rho, \pi^{\otimes t}),
  \end{align*}
  where    $\pi^{\otimes    t}$     stands    for    the    POVM
  $y\mapsto\pi_y^{\otimes t}$ and the  maximum is over ensembles
  $\rho$.
\end{dfn}

The           operational            interpretation           of
Definition~\ref{def:povm-capacity}  is   provided  by  Shannon's
noisy-channel coding theorem~\cite{Sha48}, which proves that the
capacity represents  the maximum amount of  information that can
be reliably conveyed  through POVM $\pi$ per use  of the device,
in the asymptotic limit.

Its  explicit computation  is in  general very  challenging.  To
proceed, let us first notice  that the problem can be simplified
as follows
\begin{align}
  \label{eq:infopower}
  C(\pi) = W(\pi) := \max_{\rho} I(\rho, \pi),
\end{align}
where  $W(\pi)$ is  the  {\em informational  power}~\cite{DDS11,
  OCMB11, Hol12, Hol13, SS14,  DBO14, Szy14, Dal14, Dal15, SS15}
of  POVM   $\pi$,  and  its   additivity  has  been   proved  in
Ref.~\cite{DDS11}.

A further simplification in  the calculation of $W(\pi)$ follows
from the fact  that, without loss of generality,  the maximum in
Eq.~\eqref{eq:infopower} can  be taken  over pure  ensembles, as
proved in Ref.~\cite{DDS11}.

We now introduce the following important preliminary result.

\begin{lmm}
  \label{lmm:capacity}
  The capacity $C(\pi)$ of any  POVM $y \mapsto d q_y \hat\pi_y$
  is upper bounded by
  \begin{align}
    \label{eq:capacity}
    C(\pi) \le & \max_\phi \kl{q^{\phi}}{ q}  \nonumber\\
    = &  \ln d  - d  \min_\phi \int  \dif \mu(y)  q_y \eta\left(
      \bra{\phi} \hat\pi_y \ket{\phi} \right),
  \end{align}
  where $q^{\phi}$  denotes the  probability density  $y \mapsto
  dq_y\bra{\phi}\hat\pi_y\ket{\phi}$ and $\eta(x) := - x \ln x$.
  The inequality is tight if $\openone/d \in \conv(\{ \hat\phi_x
  \})$  where $\{  \hat\phi_x \}$  consists of  the states  that
  optimize  Eq.~\eqref{eq:capacity}  and   $\conv$  denotes  the
  convex hull.
\end{lmm}

\begin{proof}
  For   any    pure   ensemble    $x\mapsto   \phi_x    :=   p_x
  \ketbra{\hat\phi_x}{\hat\phi_x}$,  let  $\sigma :=  \int  \dif
  \mu(x) p_x \hat\phi_x$ be the average state. Then one has
  \begin{align*}
    & I(\phi_x,  \pi_y) \\  = &  \kl{d p_x  q_y \bra{\hat\phi_x}
      \hat\pi_y   \ket{\hat\phi_x}}{   d  p_x   q_y   \Tr[\sigma
      \hat\pi_y] }  \\ = & d  \iint \dif \mu(x) \dif  \mu(y) p_x
    q_y \bra{\hat\phi_x}  \hat\pi_y \ket{\hat\phi_x}  \ln \frac{
      \bra{\hat\phi_x} \hat\pi_y  \ket{\hat\phi_x} } {\Tr[\sigma
      \hat\pi_y]}.
  \end{align*}
  Since one has  
  \begin{align*}
    \ln  \frac{  \bra{\hat\phi_x} \hat\pi_y  \ket{\hat\phi_x}  }
    {\Tr[\sigma  \hat\pi_y]}  =  \ln \left(  d  \bra{\hat\phi_x}
      \hat\pi_y  \ket{\hat\phi_x} \right)  -  \ln  \frac {d  q_y
      \Tr[\sigma \hat\pi_y]} {q_y},
  \end{align*}
  it follows that
  \begin{align*}
    &  I(\phi_x,  \pi_y)  \\
    =    &   \kl{d    p_x    q_y   \bra{\hat\phi_x}    \hat\pi_y
      \ket{\hat\phi_x}}  {p_x  q_y}  -  \kl  {d  q_y  \Tr[\sigma
      \hat\pi_y]} {q_y}.
  \end{align*}
  From the non-negativity of the relative entropy one has
  \begin{align*}
    &  I(\phi_x,  \pi_y)  \\
    \le   &    \kl{d   p_x   q_y    \bra{\hat\phi_x}   \hat\pi_y
      \ket{\hat\phi_x}}
    {p_x q_y} \\
    = & d \iint \dif \mu(x) \dif \mu(y) p_x q_y \bra{\hat\phi_x}
    \hat\pi_y \ket{\hat\phi_x} \ln \left( d \bra{\hat\phi_x}
      \hat\pi_y \ket{\hat\phi_x} \right) \\
    \le   &   \max_x   \kl{d  q_y   \bra{\hat\phi_x}   \hat\pi_y
      \ket{\hat\phi_x}}{q_y},
  \end{align*}
  where  the last  inequality  follows from  upper bounding  the
  average   with  the   largest  element.    Notice  that   both
  inequalities  are   saturated  iff  $\Tr[\sigma\hat\pi_y]=1/d$
  which is fulfilled whenever there exists a probability density
  $x\mapsto p_x$  such that  $\sigma =  \openone/d$ and  all the
  states $\hat\phi_x$ optimize Eq.~\eqref{eq:capacity}.
\end{proof}

Notice   that,   as   a    trivial   consequence   of   Holevo's
theorem~\cite{Hol73}, the  quantity $\ln d -  C(\pi)$, for which
Lemma~\ref{lmm:capacity}   provides  a   lower  bound,   can  be
interpreted as a measure of  how suboptimal measurement $\pi$ is
for communication tasks.

The  maximization  of  the  communication capacity  can  now  be
significantly simplified in  some cases.  Indeed, let  us set $a
:=  \min_{y,\psi}   \bra{\psi}  \hat\pi_y  \ket{\psi}$,   $b  :=
\max_{y,\psi}  \bra{\psi} \hat\pi_y  \ket{\psi}$, and  let $r  :
[a,b]\to\mathbb R$  be the  Hermite interpolating  polynomial of
$\eta$ such that the assumptions  of Lemma \ref{lmm:anna} in the
Appendix  are fulfilled  (and  so $r$  interpolates $\eta$  from
below). Then
\begin{align}
  \label{eq:hermite1}
  \int \dif \mu(y) p_y \eta(\bra{\phi} \hat\pi_y \ket{\phi})\geq
  \int  \dif\mu(y)  p_y  r(\bra{\phi} \hat\pi_y  \ket{\phi})  =:
  Q(\phi),
\end{align}
with equality if and only if  $\phi$ is such that all the points
of  interpolation $x_i$  are of  the form  $\bra{\phi} \hat\pi_y
\ket{\phi}$   (then  necessarily   the  set   of  all   overlaps
$\{\bra{\phi}  \hat\pi_y  \ket{\phi}\}_{y\in   A}$  has  to  be
finite). In consequence, we get
\begin{align}
  \label{eq:hermite2}
  C(\pi) \le \ln d - \min_\phi Q(\phi).
\end{align}
Let us now assume that we  have some predictions about the state
that    minimizes    the    lhs    of    the    inequality    in
Eq.~\eqref{eq:hermite1}  and   let  $\phi_0$  be   the  supposed
minimizer. Then, if the points of interpolation are chosen to be
$\{\bra{\phi_0} \hat\pi_y \ket{\phi_0}\}_{y\in  A}$, in order to
show that $\phi_0$  is indeed a minimizer, it is  enough to show
that  $\phi_0$  minimizes  $Q$.   In  particular,  this  becomes
trivial whenever $Q(\phi)$ is  constant.  In consequence, we get
the   equality   in   Eq.~\eqref{eq:hermite2}  exactly   as   in
Lemma~\ref{lmm:capacity}.  The remainder of this work is devoted
to the  quantification of  $C(\pi)$ through the  optimization of
$Q(\phi)$, for the specific case of $t$ designs.

\subsection{Main results}
\label{sec:results}

Our first main result is an upper bound on the capacity $C(\pi)$
of any $t$ design POVM $\pi$, as a function of the dimension and
of   the   indices    of   coincidence   $\gamma_k(\chi)$   (or,
equivalently, of the moments $\mu_k(\pi)$), for $t \in [1, 5]$.

\begin{thm}
  \label{thm:bounds}
  The classical capacity  $C(\pi)$ of any $t$  design POVM $\pi$
  is upper bounded by $C(\pi) \le C_t$, where
  \begin{align*}
    \begin{cases}
      C_1 = & \ln d \\
      C_2 = & \ln{d} + \ln \frac{\gamma_2}{\gamma_1}, \\
      C_3  =  &  \ln  d   +  \frac{d  (\gamma_1  -  \gamma_2)^2}
      {\gamma_1-2 \gamma_2+ \gamma_3} \ln
      \frac{\gamma_2-\gamma_3}{\gamma_1-\gamma_2}, \\
      C_4  =   &  \ln  d   +  \frac12  \ln   \frac{\gamma_3^2  -
        \gamma_2\gamma_4}{\gamma_2^2-   \gamma_1   \gamma_3} \\ & +
      \frac{d  (\gamma_1^2   \gamma_4  -  3   \gamma_1  \gamma_2
        \gamma_3  +  2   \gamma_2^3)}{2  \sqrt{\Delta_4}}   \ln
      \frac{\gamma_2\gamma_3     -    \gamma_1     \gamma_4    +
        \sqrt{\Delta_4}}{\gamma_2\gamma_3 - \gamma_1 \gamma_4 -
        \sqrt{\Delta_4}},
    \end{cases}
  \end{align*}
  with  $\Delta_4   =  -3\gamma_2^2  \gamma_3^2  +   4  \gamma_1
  \gamma_3^3+  4  \gamma_2^3  \gamma_4  -  6  \gamma_1  \gamma_2
  \gamma_3 \gamma_5 + \gamma_1^2 \gamma_4^2$.
\end{thm}

\begin{proof}
  Let  $0=x_0<x_1<\ldots<x_{\ceil{\frac{t}2}}\leq   1$  and  let
  $x_{\ceil{\frac{t}2}}  =  1$ iff  $t$  is  odd.  According  to
  Lemma~\ref{lmm:anna}, the polynomial $r(x)  := \sum_i a_i x^i$
  of degree at  most $t$ such that $r(x_k) =  \eta(x_k)$ for any
  $k \in  [0, \ceil{\frac{t}2}]$ and $r'(x_k)  = \eta'(x_k)$ for
  any $k  \in [1,  \ceil{\frac{t-1}2}]$ is  such that  $r(x) \le
  \eta(x)$    for   any    $x    \in    [0,   1]$.     Therefore
  Eq.~\eqref{eq:capacity} becomes
  \begin{align*}
    C(\pi) \le C_t := \ln d - d \sum_{i = 1}^t a_i \gamma_i,
  \end{align*}
  where  we  used  the  fact   that  the  index  of  coincidence
  $\gamma_i$ is  independent of the  choice of $\phi$  for $i\le
  t$.
  
  In order to upper bound $C(\pi)$,  we express $\{ a_i \}$ as a
  function of $\{  x_k \}$ and minimize $C_t$ over  $\{ x_k \}$.
  $C_1$  is an  immediate  result  as it  does  not involve  any
  optimization.  Upon  denoting with  $\{ x_k^* \}$  the optimal
  solution one has
  \begin{align*}
    \begin{cases}
      x_1^*=   \frac{\gamma_2}{\gamma_1}, & \textrm{for } t = 2,\\
      x_1^*= \frac{\gamma_2 - \gamma_3}{\gamma_{1} - \gamma_2},
      & \textrm{for } t = 3, \\
      x_{1,2}^* = \frac{\gamma_1 \gamma_4 - \gamma_2\gamma_3 \pm
        \sqrt{\Delta_{4}}} {2(\gamma_1  \gamma_3 - \gamma_2^2)},
      & \textrm{for } t = 4.
  \end{cases}
  \end{align*}
\end{proof}

The expression for  $C_5$ is too lengthy to  be reproduced here,
but the  derivation goes along the  same lines as that  of $C_4$
with
\begin{align*}
  x_{1,2}^*=  \frac{  \gamma_{2}\gamma_{3}  -  \gamma_{3}^{2}  -
    \gamma_{1}\gamma_{4}      +      \gamma_{3}\gamma_{4}      +
    \gamma_{1}\gamma_{5}     -      \gamma_{2}\gamma_{5}     \pm
    \sqrt{\Delta_5}}   {2(\gamma_{2}^{2}   +  \gamma_{3}^{2}   +
    \gamma_{1}   (\gamma_{4}   -    \gamma_{3})   -   \gamma_{2}
    (\gamma_{3} + \gamma_{4}))},
\end{align*}
and
\begin{align*}
\Delta_5   =  &   \left(\gamma_{3}   (\gamma_3   -  \gamma_2   -
  \gamma_{4})+  \gamma_{1} (\gamma_{4}-\gamma_{5})  + \gamma_{2}
  \gamma_{5}\right)^{2}    \\    &   +    4\left(\gamma_{2}^{2}+
  \gamma_{3}^{2}   +    \gamma_{1}   (\gamma_{4}-\gamma_{3})   -
  \gamma_{2}(\gamma_{3}+\gamma_{4})\right)  \\  & \times  \left(
  \gamma_{4}   (\gamma_{2}  -   \gamma_{4})  -\gamma_{3}^{2}   -
  \gamma_{2}\gamma_{5}     +     \gamma_{3}    (\gamma_{4}     +
  \gamma_{5})\right).
\end{align*}

Notice that  the optimization over $\{  x_k \}$ in the  proof of
Theorem~\ref{thm:bounds}   is  over   $\ceil{\frac{t-1}2}$  real
parameters,  and becomes  cumbersome  for $t$  larger than  $5$,
namely for $3$ parameters or more.

Notice    also    that,    as   expected,    the    bounds    in
Theorem~\ref{thm:bounds}    reduce    to    those    given    in
Ref.~\cite{Dal15}   for  projective   $t$  designs,   i.e.  when
$\mu_k(\chi) = 1$ for all $k \in [1, t]$.

Our  second  main  result  is the  derivation  of  the  capacity
$C(\mathcal{D}_\lambda(\pi))$  for  the depolarized  version  of
several         $t$-design        POVMs:         $2$-dimensional
SIC~\cite{Dav78,DDS11,OCMB11,SS14,DBO14,Szy14,Dal14,Dal15}
(tetrahedron),  complete MUB~\cite{SS14,Dal15}  (octahedron), or
icosahedron~\cite{SS14,Dal15};                   $3$-dimensional
SIC~\cite{DBO14,Szy14}   or   complete   MUB~\cite{Dal14,Dal15};
$8$-dimensional Hoggar SIC~\cite{SS15}; $d$-dimensional anti-SIC
or uniform rank-one~\cite{JRW94} POVM.

\begin{thm}
  \label{thm:examples}
  The  classical capacity  $C(\mathcal{D}_\lambda(\pi))$ of  the
  depolarized POVM  $\pi$, where $\pi$ is  a $2$-dimensional SIC
  (tetrahedron), complete  MUB (octahedron), or  icosahedron, or
  $3$-dimensional SIC or complete MUB, or $8$-dimensional Hoggar
  SIC, or $d$-dimensional anti-SIC, or the uniform rank-one POVM
  is given by
  \begin{align*}
    \begin{cases}
      C_\textrm{tetra}  =   &  \ln2  -  \frac{     \eta  \left(
          \frac{1-\lambda}2 \right) + 3 \eta \left(
          \frac{3+\lambda}6 \right)}2,\\
      C_\textrm{octa}  =  &  \ln  2   -  \frac{   \eta  \left(
          \frac{1-\lambda}2 \right) + 4 \eta \left( \frac12
        \right) +  \eta \left( \frac{1+\lambda}2 \right)}3,\\
      C_\textrm{icosa}   =   &    \ln2   -   \frac{\eta   \left(
          \frac{1-\lambda}2    \right)   +    5   \eta    \left(
          \frac{5-\sqrt{5}\lambda}{10} \right)  + 5  \eta \left(
          \frac{5+\sqrt{5}\lambda}{10} \right) + \eta
        \left( \frac{1+\lambda}2 \right)}6,\\
      C_{3\textrm{SIC}}  =  &  C_{3MUB}  = \ln3  -  \eta  \left(
        \frac{1-\lambda}3 \right) -
      2 \eta \left( \frac{2+\lambda}6 \right),\\
      C_{\textrm{Hoggar}}  =  &  \ln   8  -  \frac{7\eta  \left(
          \frac{1-\lambda}8 \right) + 9 \eta \left(
          \frac{9+7\lambda}{72} \right)}2,\\
      C_\textrm{anti-SIC}  =  &  \ln  d -  \frac1d  \eta  \left(
        \frac{1-\lambda}d \right) - \frac{d^2-1}d \eta \left(
        \frac{d^2-1+\lambda}{d(d^2-1)}     \right),\\
      C_\textrm{uniform}   =  &   \ln(1-\lambda)  +   \lambda  +
      \frac{d\lambda^2      \hg       \left(1,      1;      d+2;
          -\frac{d\lambda}{1-\lambda}\right)}{(d+1)(1-\lambda)},
    \end{cases}
  \end{align*}
  where $\hg$ denotes the hypergeometric function~\cite{OLBC10}.
  Moreover, the  optimal ensembles for the  depolarized versions
  of  $\pi$ are  exactly  the same  as for  $\pi$  and they  all
  average to the maximally mixed state.
\end{thm}

\begin{proof}
  It is well known that all the rank-one POVMs $\pi$ included in
  the theorem are projective $t$ designs,  with $t = 3$ if $\pi$
  is a  $2$-dimensional complete  MUB, $t  = 5$  if $\pi$  is an
  icosahedron,   $t  =   \infty$   if  $\pi$   is  the   uniform
  distribution, and  $t =  2$ otherwise. It  is easy  to observe
  that the $d$-dimensional anti-SIC  is a mixed 2-design.  Thus,
  due  to  Eq.~\eqref{eq:affine}, $\mathcal{D}_\lambda(\pi)$  is
  also a $t$ design.

  Let us  first discuss the  case where  $t$ is finite.   Let us
  recall  what the  optimal ensembles  are for  POVMs $\pi$,  as
  given  in Refs.~\cite{SS14,DBO14,Szy14,Dal14,Dal15,SS15}.   We
  get the dual tetrahedron for tetrahedral SIC and 'dual' Hoggar
  lines  for   the  Hoggar   SIC,  (the  same)   octahedron  and
  icosahedron    for   octahedral    and   icosahedral    POVMs,
  respectively, the Hesse SIC for 3-dimensional complete MUB and
  finally an  orthonormal basis consisting of  states orthogonal
  to exactly $3$ different $\pi_y$ for generic 3-dimensional SIC
  or the complete MUB for the Hesse SIC (note that the existence
  of  such  a basis  was  proved  in  Ref. \cite{DBBA}  for  the
  group-covariant  (Weyl-Heisenberg)  case while  just  recently
  Ref. \cite{HS16} gave the proof that there are no other qutrit
  SICs). We shall  also show that for the  anti-SIC, the optimal
  ensemble is  given by the  SIC itself.  Notice that  all these
  ensembles average to the maximally mixed state.
  
  For any state $\phi$ of such an ensemble one has
  \begin{align*}
    \bra{\phi}  \hat\pi_y \ket{\phi}  =
    \begin{cases}
      x_0  \times  m_0  ,\\
      \dots\\
      x_l \times m_l
    \end{cases},
  \end{align*}
  where  $x_1 \le  \dots  \le x_l$,  $x_1  = 0$  and  $x_l =  1$
  whenever  $t$ is  odd.   Here the  notation  $x_k \times  m_k$
  denotes the value $x_k$ with multiplicity $m_k$.  Moreover, in
  all  these   cases  there  are   exactly  $\ceil{\frac{t}2}+1$
  different         values         of        $x_k$,         thus
  $l:=\ceil{\frac{t}2}$. Therefore one has
  \begin{align*}
    \bra{\phi}  \mathcal{D}_\lambda(\hat\pi) \ket{\phi} =
    \begin{cases}
      \lambda
      x_0 + (1-\lambda) \frac1d \times  m_0,\\
      \dots\\
      \lambda x_{\ceil{\frac{t}2}} +
      (1-\lambda) \frac1d \times m_{\ceil{\frac{t}2}}.
    \end{cases}
  \end{align*}
  Notice that
  \begin{align*}
    \begin{cases}
      \lambda  x_1  +  (1-\lambda)  \frac1d  :=  \min_{\phi,  y}
      \Tr[\phi
      \mathcal{D}_\lambda(\hat\pi_y)],\\
      \lambda  x_{\ceil{\frac{t}2}+1} +  (1-\lambda) \frac1d  :=
      \max_{\phi, y} \Tr[\phi \mathcal{D}_\lambda(\hat\pi_y)] \;
      \textrm{ if $t$ is odd} ,
  \end{cases}
  \end{align*}
  Therefore,  according  to  Lemma~\ref{lmm:anna},  the  Hermite
  interpolating polynomial  $r$ such  that $r(x_k)  = \eta(x_k)$
  for  any   $k  \in  [0,  \ceil{\frac{t}2}]$   and  $r'(x_k)  =
  \eta'(x_k)$ for  any $k  \in [1, \ceil{\frac{t-1}2}]$  is such
  that  $\deg r\le  t$ and  $r(x) \le  \eta(x)$ for  any $x  \in
  [\frac{1-\lambda}d,     \lambda    +     (1-\lambda)\frac1d]$.
  Therefore, the average  of this polynomial has  to be constant
  and    according   to    our   remarks    at   the    end   of
  Sec.~\ref{sec:capacity},         the        maximum         in
  Eq.~\eqref{eq:capacity} is attained by $\phi$.

  Let  us now  discuss the  case  $t =  \infty$.  The  statement
  follows   by    expanding   $\eta(x)   =   -x    \ln   x$   in
  Lemma~\ref{lmm:capacity} in a Taylor  series around $\lambda +
  \frac{1-\lambda}d$, applying  the binomial  theorem, replacing
  $x        =        \lambda|\bra{\phi}U_g\ket{\phi}|^2        +
  \frac{1-\lambda}{d}$,  and using  the  identity  $\int \dif  g
  |\bra{\phi}  U_g \ket{\phi}|^{2k}  =  {d  + k  -  1 \choose  k
  }^{-1}$. Notice  that the series  for $\eta(x)$ has  radius of
  convergence $\lambda + \frac{1-\lambda}d$, which is also equal
  to $b  := \max_\psi  \bra{\psi} \mathcal{D}_\lambda(\hat\pi_y)
  \ket{\psi}$, therefore  $\eta(x)$ is equivalent to  its series
  in $[a, b]$.
\end{proof}

Notice     that,    as     expected,    the     capacities    in
Theorem~\ref{thm:examples}    reduce   to    those   given    in
Refs.~\cite{DDS11,SS14,DBO14,Szy14,Dal14,Dal15,SS15,JRW94}   for
projective $t$ designs when one takes the limit $\lambda \to 1$.

The results  of Theorems~\ref{thm:bounds} and~\ref{thm:examples}
are  represented   in  Figs.~\ref{fig:qubit},  \ref{fig:qutrit},
and~\ref{fig:hoggar},  in the  $2$-,  $3$-, and  $8$-dimensional
cases, respectively.

\begin{figure}[h]
  \includegraphics[width=\columnwidth]{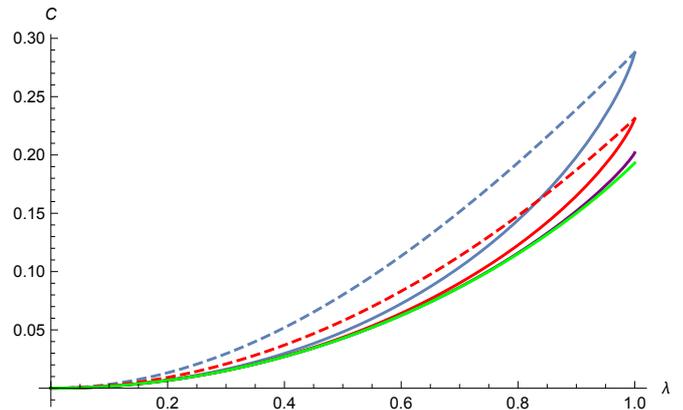}
  \caption{(Color  online) Classical  capacities of  qubit mixed
    $t$-design POVMs as a function of the depolarizing parameter
    $\lambda$.  Dashed lines, from  top to bottom, represent the
    upper bounds on  the capacity of the  depolarized version of
    any projective  $2$ design (blue)  and $3$ design  (red), as
    given by Theorem~\ref{thm:bounds}.  Solid lines, from top to
    bottom, represent the capacities  of the depolarized version
    of any SIC (blue), complete MUB (red), icosahedron (purple),
    and    uniform   distribution    (green),   as    given   by
    Theorem~\ref{thm:examples}.}
  \label{fig:qubit}
\end{figure}

\begin{figure}[h]
  \includegraphics[width=\columnwidth]{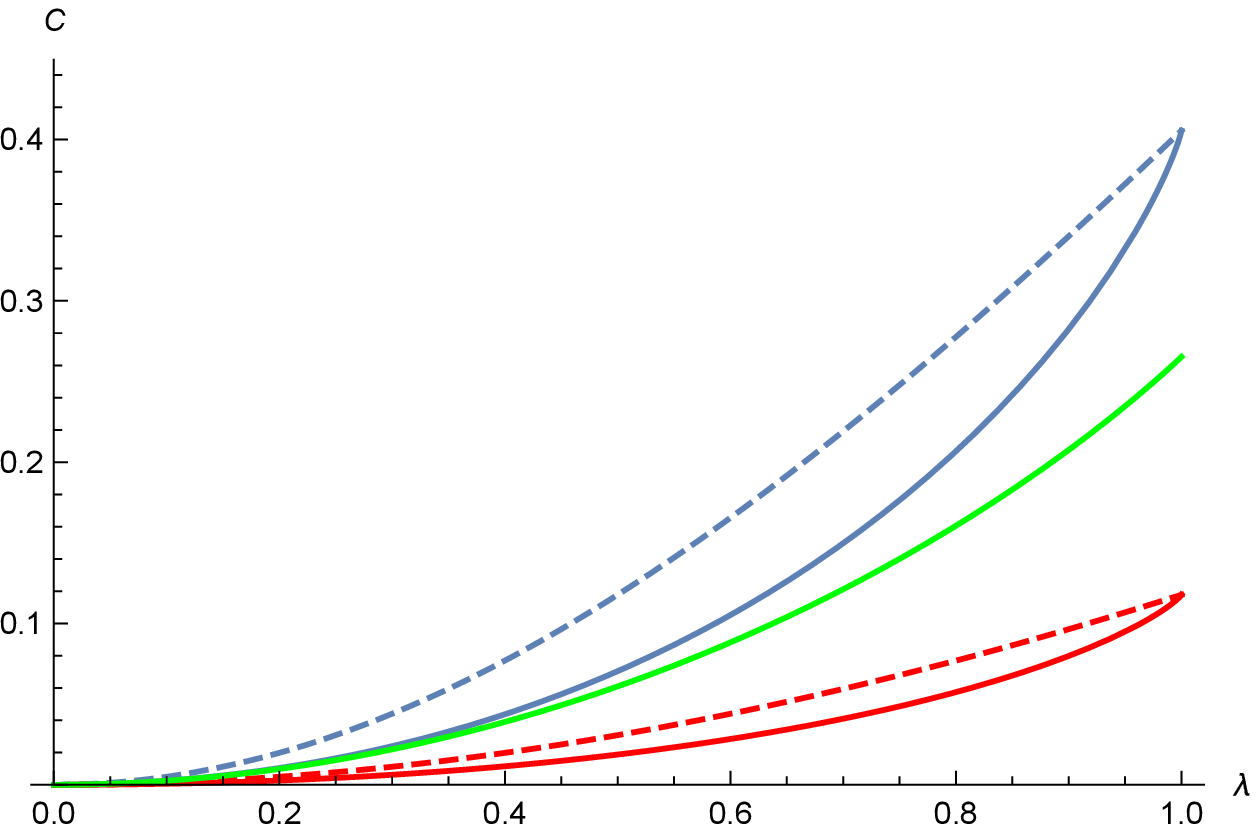}
  \caption{(Color online)  Classical capacities of  qutrit mixed
    $t$-design POVMs as a function of the depolarizing parameter
    $\lambda$.  Dashed lines, from  top to bottom, represent the
    upper bounds on  the capacity of the  depolarized version of
    any projective $2$ design (blue) and any anti projective $2$
    design (red),  as given by  Theorem~\ref{thm:bounds}.  Solid
    lines, from top  to bottom, represent the  capacities of the
    depolarized  version  of any  SIC  and  MUB (blue),  uniform
    distribution  (green),  and  anti-SIC  (red),  as  given  by
    Theorem~\ref{thm:examples}.}
  \label{fig:qutrit}
\end{figure}

\begin{figure}[h]
  \includegraphics[width=\columnwidth]{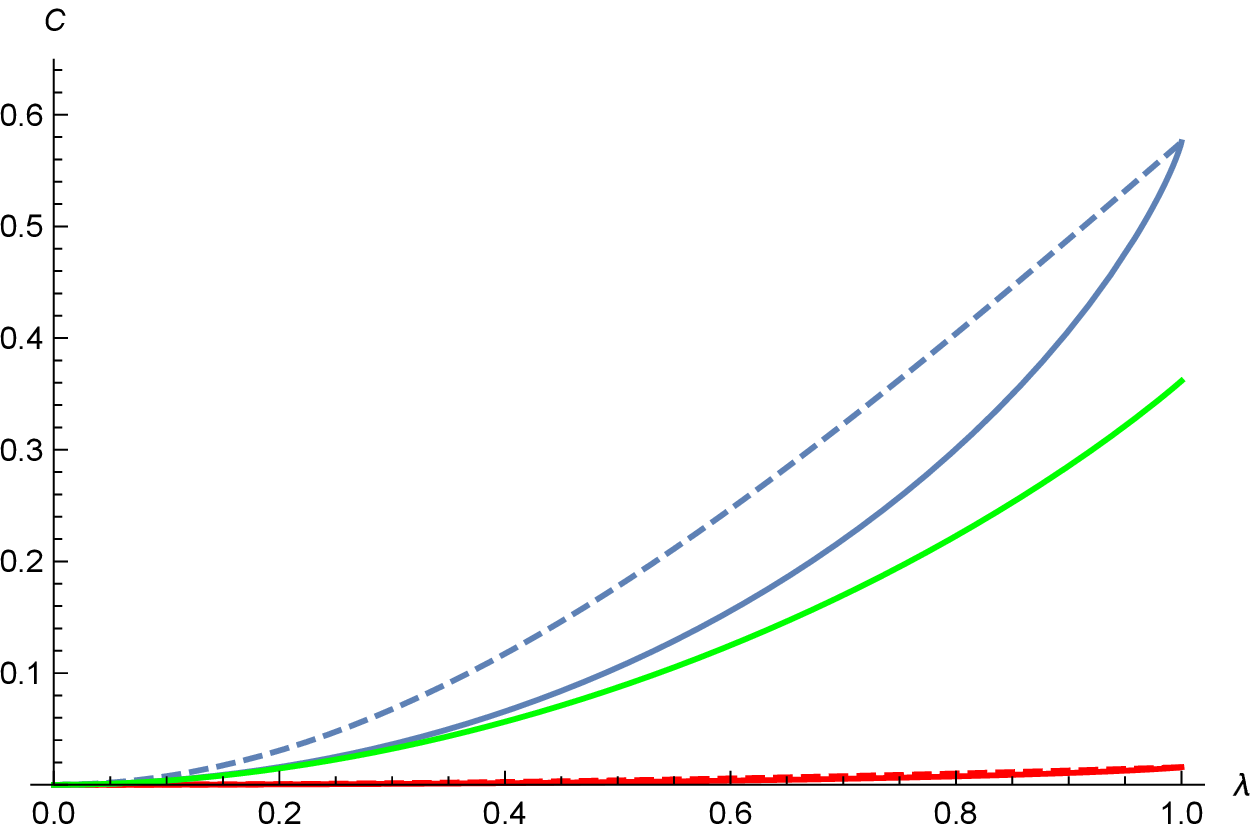}
  \caption{(Color     online)     Classical    capacities     of
    $8$-dimensional mixed $t$-design POVMs  as a function of the
    depolarizing parameter $\lambda$.  Dashed lines, from top to
    bottom, represent  the upper bounds  on the capacity  of the
    depolarized version of any  projective $2$ design (blue) and
    any  anti   projective  $2$   design  (red),  as   given  by
    Theorem~\ref{thm:bounds}.  Solid lines,  from top to bottom,
    represent the  capacities of the depolarized  version of any
    Hoggar  SIC   (blue),  uniform  distribution   (green),  and
    anti-SIC (red), as given by Theorem~\ref{thm:examples}.}
  \label{fig:hoggar}
\end{figure}

Finally,  let  us  notice  that the  capacity  $C(\pi)$  of  any
measurement $\pi$  can be expressed  in terms of  the well-known
accessible information~\cite{LL66,  Lev92, Lev95,  Lev96, Hol73,
  Bel75a,  Bel75b,  Dav78,  JRW94, DW04}  $A(\rho)  :=  \max_\pi
I(\rho, \pi)$ of ensemble $\rho$,  where the maximum is over any
POVM $\pi$.   The relation is  given by $C(\pi)  = \max_{\sigma}
A(\sigma^{1/2}  \pi \sigma^{1/2})$,  where by  $\sigma^{1/2} \pi
\sigma^{1/2}$  we  mean  the  ensemble  $y  \mapsto  \frac{1}{d}
\sigma^{1/2} \pi_y  \sigma^{1/2}$ and the maximum  is taken over
all states $\sigma$.  Therefore, in  general one has $W(\pi) \ge
A(\frac1d \pi)$, and  thus the upper bound  on $C(\pi)$ provided
in  Thm.~\ref{thm:bounds}  also   upper  bounds  the  accessible
information $A(\frac1d  \pi)$.  Moreover, whenever  the ensemble
attaining the capacity for POVM  $\pi$ averages to the maximally
mixed state,  one has that~\cite{DDS11} the  informational power
$W(\pi) =  A(\frac1d \pi)$.  Therefore, the  expressions for the
capacities   in  Thm.~\ref{thm:examples}   also  represent   the
accessible  information $A(\frac1d  \pi)$  of the  corresponding
ensembles.

\section{Conclusions and Outlook}
\label{sec:conclusion}

In this work we introduced mixed quantum $t$ designs as the most
general arbitrary-rank extension of  projective $t$ designs that
preserves  indistinguishability  from the  uniform  distribution
given $t$ copies.   We addressed the problem  of quantifying the
communication  capacity  of  mixed $t$  design  measurements  by
deriving upper bounds on such a quantity for any $t \in [1, 5]$.
We refined our  results by providing a  closed-form solution for
the  communication   capacity  of  several  mixed   $t$  designs
measurements, including  the depolarized  version of:  any qubit
and qutrit SIC and MUBs,  any qubit icosahedral measurement, any
Hoggar SIC, any anti-SIC, and the uniform distribution.

One   might   conjecture   that  the   quantification   of   the
communication   capacity   of   mixed  $(t   =   \infty)$-design
measurements  would shine  new  light on  the  problem of  lower
bounding the  communication capacity of mixed  quantum POVMs, in
the  same   way  as   the  quantification~\cite{JRW94}   of  the
accessible  information  of  projective $(t  =  \infty)$-designs
allowed for the  computation of a lower bound  on the accessible
information of ensembles of pure states.

\section*{Acknowledgments}
The authors are  grateful to Jung Jun Park for  pointing out the
relation between the  index of coincidence of  mixed $t$ designs
and      Bell       polynomials,      as       expressed      in
Eq.~\eqref{eq:coincidence}.   M.D.   is  grateful  to  Francesco
Buscemi, Massimiliano F.  Sacchi, and Vlatko Vedral for valuable
discussions.  This  work is  supported by Singapore  Ministry of
Education    Academic    Research    Fund    Tier    3    (Grant
No. MOE2012-T3-1-009).

\appendix

\section{Hermite interpolation}
\label{sec:hermite}

For the completeness of our reasoning we add a brief description
of the optimization method based on Hermite interpolation, first
introduced  in   this  context   and  discussed  in   detail  in
Ref.~\cite{SS14}. Let us first recall the well known formula for
the Hermite interpolation error.
\begin{lmm}
  Let $a  \le x_0  < \dots  < x_{m-1}  \le b$  be reals  and let
  $\{j_{i}\}_{i=0}^{m-1}$ be  positive integers.   Let $\eta(x)$
  be  a  real  function  with   at  least  $t  +  1$  continuous
  derivatives  on   $[a,b]$  and  let  $r(x)$   be  its  Hermite
  interpolating polynomial,  namely the polynomial of  degree at
  most $t :=  \sum_{i} j_{i} - 1$ that agrees  with $\eta(x)$ at
  $x_{i}$ up to derivative of order $j_{i}-1$ for all $i$, i.e.
  \begin{align*}
    r^{(k)}(x_{i}) = \eta^{(k)}(x_{i}), \qquad  \forall 0 \leq k
    \leq j_{i}-1, \ \ 0 \leq i \leq m -1.
  \end{align*}
  Then  for  any $x  \in  [a,b]$  there  exists $x'$  such  that
  $\min(x,x_0) <x' < \max(x,x_{m-1})$ and
  \begin{align*}
    \eta(x)    -    r(x)    =    \frac{\eta^{(t+1)}(x')}{(t+1)!}
    \prod_{i=0}^{m-1} (x-x_{i})^{j_{i}}.
  \end{align*}
\end{lmm}

\begin{proof}
  See, e.g. Ref.~\cite{SB02}.
\end{proof} 

In this work we set $\eta(x) := - x \ln(x)$. Next lemma provides
sufficient  conditions for  the  polynomial  $r$ to  interpolate
$\eta$ from below.
 
\begin{lmm}
  \label{lmm:anna}
  Let $m\geq 2$, $x_0=a$ and $j_0=1$. Then whenever
  \begin{enumerate}[a.]
  \item $x_{m-1}=b$, $j_{m-1} = 1$  and $j_i = 2$ for $0<i<m-1$,
    or
  \item $x_{m-1}<b$, and $j_i = 2$ for $0<i$,
  \end{enumerate}
  one has $\eta(x) - r(x) \geq 0$ for any $x \in [a, b]$.
\end{lmm}

\begin{proof}
  The proof was first derived in Ref.~\cite{SS14}. We report it
  here for clarity.
  \begin{enumerate}[a.]
  \item In this case one has
    \begin{align*} 
      \prod_{i=0}^{m-1}   (x-x_{i})^{j_{i}}    =   (x-a)   (x-b)
      \prod_{i=0}^{m-1} (x-x_{i})^{2} \leq 0.
    \end{align*}
    Likewise, $t$ is odd and so $\frac{\eta^{(t+1)}(x')}{(t+1)!}
    \le  0$,  since  all   even  derivatives  of  $\eta(x)$  are
    negative.
  \item In this case
    one has
    \begin{align*}
      \prod_{i=0}^{m-1}      (x-x_{i})^{j_{i}}      =      (x-a)
      \prod_{i=0}^{m-1} (x-x_{i})^{2} \geq 0.
    \end{align*}
    Likewise, $t$  is even  and $\frac{\eta^{(t+1)}(x')}{(t+1)!}
    \ge  0$, since  all  odd (greater  than  1) derivatives  of
    $\eta(x)$ are positive.\qedhere
  \end{enumerate}
\end{proof}

\end{document}